\documentclass[aps,prd,preprint,superscriptaddress,showpacs,floatfix,nobibnotes]{revtex4-1}

\usepackage[latin1]{inputenc}
\usepackage{graphicx}
\usepackage{amsmath}
\usepackage{geometry}   % See geometry.pdf to learn the layout options. There are lots.
\usepackage{amsfonts}
\usepackage{amssymb}
\usepackage{amsthm}
\usepackage{slashed}
\usepackage{mathrsfs}
\usepackage{enumerate}
\usepackage{color}
\usepackage[yyyymmdd,hhmmss]{datetime}

\numberwithin{equation}{section}

\newcommand{\bea}{\begin{eqnarray}}
\newcommand{\eea}{\end{eqnarray}}

\newcommand{\TT}{\sigma}

\newcommand{\BbbR}{\mathbb{R}}

\newcommand{\BbbN}{\mathbb{N}}

\DeclareMathOperator{\arccot}{arccot}
\DeclareMathOperator{\Ci}{Ci}
\DeclareMathOperator{\Cin}{Cin}
\DeclareMathOperator{\Si}{Si}

\DeclareMathOperator{\Realpart}{Re}
\DeclareMathOperator{\Imagpart}{Im}

\theoremstyle{plain}
\newtheorem{thm}{Theorem}[section]
\newtheorem{lem}[thm]{Lemma}
\newtheorem{prop}[thm]{Proposition}

\begin{document}
\title{Smooth and sharp creation of a pointlike source\\ 
for a $(3+1)$-dimensional quantum field}

\author{L.~J. Zhou}
\email[]{zhoulingjunjeff@gmail.com} 
\affiliation{Department of Physics, Brandon University, Brandon, Manitoba, R7A 6A9 Canada}
\affiliation{Winnipeg Institute for Theoretical Physics, Winnipeg, Manitoba}
\affiliation{Department of Physics, University of Winnipeg, Winnipeg, Manitoba, R3B 2E9 Canada}
\affiliation{Department of Physics and Astronomy, University of Manitoba, Winnipeg, Manitoba, R3T 2N2 Canada}

\author{Margaret E. Carrington}
\email[]{carrington@brandonu.ca} 
\affiliation{Department of Physics, Brandon University, Brandon, Manitoba, R7A 6A9 Canada}
\affiliation{Winnipeg Institute for Theoretical Physics, Winnipeg, Manitoba}

\author{Gabor Kunstatter}
\email[]{gkunstatter@uwinnipeg.ca} 
\affiliation{Winnipeg Institute for Theoretical Physics, Winnipeg, Manitoba}
\affiliation{Department of Physics, University of Winnipeg, Winnipeg, Manitoba, R3B 2E9 Canada}

\author{Jorma Louko}
\email[]{jorma.louko@nottingham.ac.uk} 
\affiliation{School of Mathematical Sciences, University of Nottingham, Nottingham NG7 2RD, United Kingdom}

\date{ {Revised  {March} 2017}}

% \date{\today}

\begin{abstract}
We analyse the smooth and sharp creation of a pointlike source for a 
quantised massless scalar field in 
$(3+1)$-dimensional Minkowski spacetime, as a model for
the breakdown of correlations that has been proposed to occur 
at the horizon of an evaporating black hole. 
The creation is implemented by a time-dependent 
self-adjointness parameter at the excised spatial origin. 
In a smooth creation, the renormalised energy density 
$\langle T_{00} \rangle$ 
is well defined away from the source, 
but it is unbounded both above and below: 
the outgoing pulse contains an infinite negative energy, 
while a cloud of infinite positive energy lingers near the fully-formed source. 
In the sharp creation limit, 
$\langle T_{00} \rangle$ 
diverges everywhere in the  {timelike future} of the creation event, 
and so does the response of an Unruh-DeWitt detector that operates 
in the  {timelike} future of the creation event. 
The source creation is 
 {significantly more singular than the corresponding 
process in $1+1$ dimensions, analysed previously, 
and} it may be sufficiently singular 
to break quantum correlations as proposed in a black hole spacetime. 
\end{abstract}


\normalsize
\maketitle

\normalsize

\section{Introduction}

In quantum field theory, it has been long known that a time dependent boundary 
condition or a time dependent metric can create particles and energy flows. 
Parker's pioneering work showed that 
a Klein-Gordon field on an expanding cosmological spacetime 
undergoes particle creation~\cite{Parker:1969au}. 
Moore showed that particle creation can be induced by varying 
the length of a cavity~\cite{Moore70}, 
while Candelas and Deutsch showed that even a single accelerating mirror can 
induce a flux of particles and energy~\cite{Candelas:1977zza}; 
this phenomenon is now known as the Dynamical (or non-stationary) Casimir Effect, 
and it was observed in 2011 using a photon analogue system~\cite{Wilson11}. 
The most celebrated example is Hawking's prediction of 
black hole radiation~\cite{hawking}, whose observation in 
analogue quantum systems may be at the threshold 
of current technology~\cite{Belgiorno:2010wn,Steinhauer:2015saa}.  

In order to reconcile the thermal character of Hawking radiation 
with fundamental unitarity of quantum theory, it has been proposed 
\cite{braunstein-et-al,Mathur:2009hf,Almheiri:2012rt,Susskind:2013tg,Almheiri:2013hfa,Hutchinson:2013kka,Harlow:2014yka}
that the horizon of a radiating black hole 
could be more singular than the conventional picture of quantum 
fields on the classical 
black hole spacetime 
suggests~\cite{birrell-davies,wald-smallbook,parker-toms-book}. 
While detailed modelling of this possible singularity 
remains elusive, 
the key proposed feature is that the singularity 
should break down correlations between the two sides of the horizon. 
A~context in which such breaking of correlations can be studied 
is quantum field theory on a fixed background spacetime. 
One way to do this is
to write down by hand a quantum state in which the correlations are 
absent~\cite{Louko:2014aba,Martin-Martinez:2015dja}. Another is 
to allow an impermeable wall to develop where 
initially there 
was none~\cite{unruh-seoultalk,Brown:2014qna,Brown:2015yma,Harada:2016kkq}. 
The purpose of the present paper is to improve 
the understanding of the latter scenario. 

When the impermeable wall is inserted quickly, a 
surprising feature emerges: 
for a massless scalar field in 
$1+1$ dimensions, 
the energy transmitted into the field diverges in 
the limit of rapid wall creation, 
but the response of an Unruh-DeWitt detector 
\cite{Unruh:1976db,DeWitt:1979}
crossing this pulse of diverging 
energy remains finite~\cite{Brown:2015yma}. 
 {The finite detector response} 
casts doubt on the ability of wall creation, however rapid, 
to break down quantum correlations sufficiently strongly to 
save unitarity in an evolving black hole spacetime. 
One limitation of the analysis in \cite{Brown:2015yma} 
is however that it was done in $1+1$ dimensions. 
Quantum fields generally become more 
singular as the spacetime dimension increases: 
would the conclusions in $3+1$ dimensions be similar? 
A second limitation is that the analysis in 
\cite{Brown:2015yma} relied on an infrared cutoff to 
eliminate the infrared ambiguity that the massless scalar field has in 
$1+1$ dimensions. Could the results in \cite{Brown:2015yma}
be an artifact of the $(1+1)$-dimensional 
infrared sickness, 
with no counterpart in $3+1$ dimensions? 

In this paper we take a first step towards 
adapting the wall creation analysis of \cite{Brown:2015yma} 
to $3+1$ dimensions, 
and answering these questions. 
We consider a massless scalar field in $(3+1)$-dimensional 
Minkowski spacetime, and we introduce at the spatial origin 
a time-dependent boundary condition that 
interpolates, over a finite interval of time, 
between ordinary Minkowski dynamics and
a Dirichlet-type condition. 
As the boundary condition is introduced at 
just one spatial point, 
the physical interpretation is now not the 
smooth creation of a wall but the 
smooth creation of a pointlike source. 
We then ask what happens to the energy transmitted into the 
field and to the response of an Unruh-DeWitt detector 
in the limit of rapid source creation. 
The answers turn out to have some similarities with the $(1+1)$-dimensional 
analysis of \cite{Brown:2015yma} but also significant 
differences. A~technical difference is that in $3+1$ dimensions there is 
no infrared ambiguity, and no infrared cutoff is needed. 
 {A~difference in physically observable quantities is
that in $3+1$ dimensions both the field's energy density
and the detector's response are more singular.}

First, we consider the energy. 
While the renormalised energy density 
$\langle T_{00}\rangle$ is well defined everywhere away from the source, 
it is bounded neither above nor below. 
In the outgoing pulse generated by the evolving source, 
$\langle T_{00}\rangle$ is unbounded below 
immediately to the future of the light cone of the point 
where the boundary condition starts to change, 
and the total energy in the pulse is negative infinity. 
After the pulse has gone, 
$\langle T_{00}\rangle$ is nonzero, and it diverges 
at $r\to0$ proportionally to $-(\ln r)/r^4$: 
a cloud of positive energy lingers near the source after the 
source is fully formed, and the total energy in this cloud is positive infinity. 
Further, at a fixed~$r$, $\langle T_{00} \rangle$ is not static, and it 
diverges at $t\to\infty$ proportionally to $\ln t$. 
In the limit of rapid source creation, 
$\langle T_{00} \rangle$ diverges everywhere 
in the  {timelike} future of the creation event. 
The source creation hence leaves in the late 
time region a large energetic memory. This memory 
has no counterpart in the $(1+1)$-dimensional 
analysis of~\cite{Brown:2015yma}. 

We note that the firewall in both the previous paper \cite{Brown:2015yma} 
and the present work is not in fact modelled by the wall or point source, 
respectively, where the boundary conditions are specified. 
Instead, these serve as the source of the firewall which itself is 
modelled by the resulting outgoing null shell of energy. 
It is for this reason that it is important to calculate the 
response of a detector passing through the outgoing shell of energy 
(i.e.\ firewall), as we do in Section~\ref{sec:UDWdetector}. 
In particular, we consider the response of a static Unruh-DeWitt detector. 
We find that the response of a detector that 
operates only in the late time region mimics 
$\langle T_{00} \rangle$ closely, 
both in the late time 
limit and in the limit of rapid 
source creation: 
in both limits, the response has a logarithmic divergence. 
We have not considered in detail the response 
of a detector that goes through the pulse emanating from the 
changing boundary condition, but the behaviour 
in the post-pulse region is already sufficient to 
establish that the response does not remain finite 
in the limit of rapid source creation.

We conclude that the rapid creation of a source makes 
the $(3+1)$-dimensional field significantly 
more singular than the corresponding event in $1+1$ dimensions; 
in particular, the response of an Unruh-DeWitt detector 
diverges in the rapid creation limit. 
These results suggest that a source creation may be able to 
model the breaking of quantum correlations 
in the way that has been proposed to happen in an evolving black 
hole spacetime 
\cite{braunstein-et-al,Mathur:2009hf,Almheiri:2012rt,Susskind:2013tg,Almheiri:2013hfa,Hutchinson:2013kka,Harlow:2014yka}. 
The persistence of large late time effects is perhaps particularly 
reminiscent of the energetic curtain scenario proposed in~\cite{braunstein-et-al}.

We begin in Section 
\ref{sec:classical}
by setting up the classical dynamics of the scalar 
field under the evolving boundary condition at the spatial origin. 
Section \ref{sec:quantised} introduces the 
quantised field and evaluates~$\langle T_{00} \rangle$.
The response of an Unruh-DeWitt detector 
is considered in Section~\ref{sec:UDWdetector}. 
Section \ref{sec:summary} gives a brief summary and discussion. 
Technical material is relegated to five appendices. 

Our metric signature is mostly minus. 
Overline denotes complex conjugation. 
A~continuous function of a real variable 
is said to be~$C^0$, 
a function that is $n \in \BbbN = \{1,2,\ldots\}$
times continuously differentiable is said to be $C^n$, 
and a function that has all derivatives 
is said to be~$C^\infty$, or smooth. 
We work in geometric units in which $\hbar = c = 1$. 

\section{Classical field\label{sec:classical}}

\subsection{Field equation and boundary condition\label{subsec:field-eq}}

We consider a real massless scalar field $\phi$ in 
$(3+1)$-dimensional Minkowski spacetime from which the spatial origin has been excised. 
Writing the metric as 
\begin{align}
ds^2 = dt^2 - {(dx^1)}^2 - {(dx^2)}^2 - {(dx^3)}^2\,, 
\label{eq:Mink-metric}
\end{align}
the field equation is 
\begin{align}
(\partial^2_t-\nabla^2)\phi = 0\,,
\label{eq:fieldeq}
\end{align}
where $\nabla^2 = \partial_{x^1}^2 + \partial_{x^2}^2 + \partial_{x^3}^2$. 
The Klein-Gordon inner product evaluated on a constant $t$ hypersurface reads 
\begin{align}
(\phi_1, \phi_2)_{KG} 
= i \int dx^1\,dx^2\,dx^3 \, 
\bigl( \overline{\phi_1} \partial_t \phi_2 
- (\partial_t\overline{\phi_1})\phi_2 \bigr) 
\,. 
\label{eq:KG-innerproduct-cartesian}
\end{align}
In the spherical coordinates, defined by 
$(x^1, x^2, x^3) = (r\sin\theta \cos\varphi,  
r\sin\theta \sin\varphi, r\cos\theta)$, 
the metric reads 
\begin{align}
ds^2 = dt^2  - dr^2 
- r^2 \left( d\theta^2 + \sin^2 \! \theta \, d\varphi^2 \right) 
\label{eq:Mink-spherical-metric}
\end{align}
and the Klein-Gordon inner product reads 
\begin{align}
(\phi_1, \phi_2)_{KG} 
= 
i \int_0^\infty r^2 \, dr \int_{S^2} d\Omega \, 
\bigl( \overline{\phi_1} \partial_t \phi_2 - (\partial_t\overline{\phi_1})\phi_2 \bigr) 
\,,
\label{eq:KG-innerproduct-spherical}
\end{align}
where $d\Omega = \sin\theta\,d\theta \,d\varphi$ 
is the volume element on unit~$S^2$. The excised spatial origin is at $r=0$. 

To specify the dynamics, we need to define $\nabla^2$ at each $t$ as a self-adjoint operator. 
After decomposition into spherical harmonics, the only freedom is in the spherically symmetric sector, 
as discussed in Appendix~\ref{app:laplaceoperator}: writing 
\begin{align}
\phi(t,r) = \frac{f(t,r)}{\sqrt{4\pi} \, r}
\,, 
\label{eq:phi-scaledto-f}
\end{align}
the eigenfunctions of $\nabla^2$ must satisfy the boundary condition 
\begin{align}
\bigl(\cos\theta(t)\bigr)\lim_{r\to0} f(t,r)
= 
L \bigl(\sin\theta(t) \bigr) \lim_{r\to0} \partial_r f(t,r)
\,,
\label{eq:robin} 
\end{align}
where
$L$ is a positive constant of dimension length, 
introduced for dimensional convenience, 
and the prescribed function $\theta(t)$, taking values in 
$[0, \pi)$, specifies at each $t$ the self-adjoint 
extension of~$\nabla^2$. We denote this extension by 
$\Delta_{\theta(t)}$. 

$\Delta_0$ coincides with the unique self-adjoint 
extension of $\nabla^2$ on $L_2(\BbbR^3)$, yielding
usual scalar field dynamics on full Minkowski space. 
For $\theta\in(\pi/2,\pi)$, 
$\Delta_\theta$ has a positive proper eigenvalue, 
which on quantisation would give a tachyonic instability. 
We therefore assume $\theta\in[0,\pi/2]$, 
in which case the spectrum of $\Delta_\theta$ consists of the negative continuum. 

We specialise to a $\theta(t)$ that interpolates between $\theta=0$ and $\theta=\pi/2$ 
over a finite interval of time. We may parametrise $\theta(t)$ as 
\begin{align}
\theta(t) = 
\begin{cases}
0 & \text{for $t \le 0$}
\ , 
\\
\arccot \! \left[\lambda L \cot\bigl(h(\lambda t)\bigr)\right] 
& \text{for $0 < t < \lambda^{-1}$}
\ , 
\\
\pi/2 & \text{for $t \ge  \lambda^{-1}$}
\ , 
\end{cases}
\label{eq:theta-param}
\end{align}
where $\lambda$ is a positive constant of dimension inverse length and 
$h: \BbbR \to \BbbR$ 
is a smooth function such that 
\begin{subequations}
\label{eq:theta-scaled}
\begin{alignat}{2}
& h(y) = 0 &\quad& \text{for $y\le0$}
\ ,  
\\
& 0 < h(y) < \pi/2 && \text{for $0 < y < 1$}
\ , 
\\
& h(y) = \pi/2 && \text{for $y\ge1$}
\ . 
\end{alignat}
\end{subequations}
Over the interval $0<t< \lambda^{-1}$, the 
boundary condition \eqref{eq:robin} 
then reads 
\begin{align}
\lim_{r\to0}
\frac{\partial_r f(t,r)}{f(t,r)} 
= 
\lambda \cot\bigl(h(\lambda t)\bigr)
\ . 
\label{eq:scaled-BC}
\end{align}
In words, this parametrisation means that 
the boundary condition interpolation takes place over time 
$\lambda^{-1}$ while the interpolation profile is determined by the 
dimensionless 
 {function~$h(y)$.
The limit of rapid
interpolation with fixed profile is that of $\lambda\to\infty$.}

\subsection{Mode functions\label{subsec:classical-modefunctions}}

As preparation for quantisation, we shall write down 
the mode solutions that reduce to the usual Minkowski modes for $t\le0$. 
As noted above, we need consider only the spherically symmetric sector. 

We work in the radial null coordinates $u := t-r$ and $v := t+r$, 
in which $t = (v+u)/2$ and $r = (v-u)/2$. 
The metric \eqref{eq:Mink-spherical-metric} becomes 
\begin{align}
ds^2 = du\,dv 
- \tfrac14 (v-u)^2 
\left( d\theta^2 + \sin^2 \! \theta \, d\varphi^2 \right) 
\,. 
\label{eq:Mink-null-metric}
\end{align}
Taking $\phi$ to be spherically symmetric, the field equation 
\eqref{eq:fieldeq} becomes 
\begin{align}
\partial_u \partial_v (r\phi) =0 
\,.   
\end{align}
We hence seek mode solutions with the ansatz 
\begin{align}
\phi_k = \frac{U_k}{\sqrt{4\pi} \, r}
\,,
\label{eq:phisubk-modes}
\end{align}
where 
\begin{align}
U_k (u,v) = \frac{1}{\sqrt{4\pi k}} \left[ e^{-ikv}+E_k(u) \right] 
\,, 
\label{eq:Umode-ansatz}
\end{align}
$k>0$, and $E_k$ is to be found. As any choice for $E_k$ satisfies the wave equation, 
the task is to determine $E_k$ so that the boundary condition \eqref{eq:robin} 
is satisfied for all $t$ and the usual Minkowski modes are obtained for $t\le0$. 

Substituting \eqref{eq:Umode-ansatz} in the boundary condition 
\eqref{eq:robin} gives for $E_k$ the ordinary differential equation 
\begin{align}
L \sin\bigl(\theta(t)\bigr) 
\frac{d}{dt}
\! \left[
e^{-ikt} - E_k(t)
\right]
&= \cos\bigl(\theta(t)\bigr) 
\! \left[
e^{-ikt} + E_k(t)
\right]
\ . 
\label{eq:E-ev-eq}
\end{align}
Writing 
\begin{align}
E_k(u) = R_{k/\lambda}(\lambda u)
\label{eq:E-intermsof-R}
\end{align}
and using~\eqref{eq:theta-param}, 
\eqref{eq:E-ev-eq} takes the dimensionless form 
\begin{align}
\sin\bigl(h(y)\bigr) 
\frac{d}{dy}
\! \left[
e^{-iKy} - R_K(y)
\right]
&= \cos\bigl(h(y)\bigr) 
\! \left[
e^{-iKy} + R_K(y)
\right]
\ , 
\label{eq:R-ev-eq}
\end{align}
where $K = k/\lambda>0$ is the dimensionless frequency and 
$y = \lambda u$. 

To solve \eqref{eq:R-ev-eq}, we introduce the auxiliary function 
\begin{align}
B(y) = 
\begin{cases}
{\displaystyle{0}}
& \text{for $y \le 0$}
\ , 
\\
{\displaystyle{\exp\left(-\int_{y}^{1}\cot\bigl(h(z)\bigr) \, dz \right)}}
& \text{for $0 < y < 1$}
\ , 
\\[1ex]
{\displaystyle{1}}
& \text{for $y \ge 1$}
\ . 
\end{cases}
\label{eq:B-sol}
\end{align}
$B(y)$ is everywhere smooth: smoothness at $y=1$ follows 
from the smoothness of $h(z)$ near $z=1$, 
and smoothness at $y=0$ is shown 
in Appendix~\ref{app:junction-differentiability}. 
For $y>0$, $B(y)$ satisfies 
\begin{align}
\frac{B'(y)}{B(y)}
= \cot\bigl(h(y)\bigr) 
\,.
\label{eq:Bscaled-diffeq}
\end{align}
It follows that the solution to \eqref{eq:R-ev-eq} is 
\begin{align}
R_K(y) = 
\begin{cases}
{\displaystyle{- e^{-iK y}}}
& \text{for $y \le 0$}
\ , 
\\
{\displaystyle{- e^{-iK y} 
- \frac{2 i K }{B(y)}
\int_0^{y} {B}(z) \, e^{-iK z} \, dz}}
& \text{for $0 < y < \infty$}
\ . 
\end{cases}
\label{eq:R-sol-alt}
\end{align}
From \eqref{eq:R-sol-alt}
and the smoothness of $B$ we see that 
$R_K(y)$ is smooth everywhere except possibly at $y=0$, and 
we verify in Appendix \ref{app:junction-differentiability} 
that $R_K(y)$ is $C^{25}$ at $y=0$. It follows that 
the mode functions are smooth everywhere except 
possibly at $r=t$, and they 
are at least $C^{25}$ at $r=t$. 

An alternative expression for $R_K(y)$ is 
\begin{align}
R_K(y) = 
\begin{cases}
{\displaystyle{- e^{-iK y}}}
& \text{for $y \le 0$}
\ , 
\\
{\displaystyle{e^{-iK y} - \frac{2}{B(y)}
\int_0^{y} {B'}(z) \, e^{-iK z} \, dz}}
& \text{for $0 < y < 1$}
\ , 
\\[1ex]
{\displaystyle{e^{-iK y} - 2 C_K}}
& \text{for $y \ge 1$}
\ , 
\end{cases}
\label{eq:R-sol-full}
\end{align}
where 
\begin{align}
C_K = 
\int^{1}_0 B'(z) \, e^{-iKz} \, dz
\,.
\label{eq:CsubK-def}
\end{align}
At $u\le0$ and $u\ge \lambda^{-1}$, the mode functions 
$\phi_k$ \eqref{eq:phisubk-modes} hence 
reduce respectively to 
\begin{align}
\phi_k(t,r) = 
\begin{cases}
{\displaystyle{- \frac{i e^{-ikt}\sin(kr)}{2 \pi \sqrt{k} \, r}}}
& \text{for $u\le0$}
\ , 
\\[2ex]
{\displaystyle{\frac{e^{-ikt} \cos(kr) - C_{k/\lambda}}{2 \pi \sqrt{k} \, r}}}
& \text{for $u\ge \lambda^{-1}$}
\ . 
\end{cases}
\label{eq:phisubk-sol-early-and-late}
\end{align}
For $u\le0$, $\phi_k(t,r)$ coincide with the usual 
Minkowski space mode functions. Evaluating the Klein-Gordon inner product 
\eqref{eq:KG-innerproduct-spherical}
on a hypersurface of constant negative 
$t$ shows that the normalisation is $(\phi_k, \phi_{k'})_{KG} = \delta(k-k')$. 
For $u\ge \lambda^{-1}$, the $r$-dependence in the numerator of $\phi_k(t,r)$
\eqref{eq:phisubk-sol-early-and-late} contains the term 
$\cos(kr)$, which 
one would expect from the boundary condition 
\eqref{eq:robin} with $\theta=\pi/2$, 
but it contains also the additive memory term $-C_{k/\lambda}$, 
which carries a recollection of how the boundary condition 
evolved from $\theta=0$ to $\theta=\pi/2$. 
 {From \eqref{eq:CsubK-def} 
we see that $C_K$ is smooth in~$K$, $C_0=1$, and} 
$C_K \to 0$ faster than any inverse power of $K$ as $K\to\infty$, 
as can be verified by repeated integration by parts~\cite{wong}. 
For fixed~$\lambda$, the memory term is hence insignificant 
at large frequencies but significant at low frequencies. 
We shall see in Section \ref{sec:quantised} that the 
memory term has a significant effect on 
the stress-energy tensor and the Wightman 
function. 

A spacetime diagram is shown in Figure~\ref{trplane-fig}, indicating 
the regions $u<0$, $0 < u < \lambda^{-1}$ and $u>\lambda^{-1}$. 

\begin{figure}[t]
\begin{center}
% {\huge{insert figure here}}
\includegraphics[width=6.5cm]{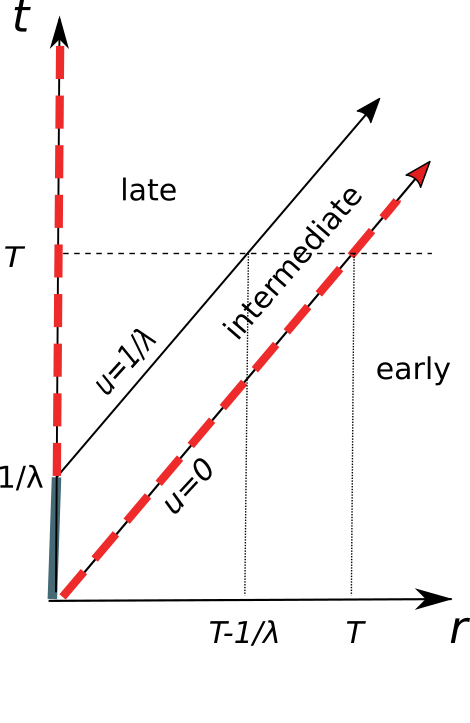}
\end{center}
\caption{Spacetime diagram of the evolving boundary condition \eqref{eq:robin} at $r=0$, 
with the angular dimensions suppressed. 
The interpolation between $\theta=0$ and $\theta=\pi/2$ at $r=0$ 
occurs over $0<t<\lambda^{-1}$ {(solid line)}, and the null cones of the 
events where the boundary condition changes  
fill the region $0<u<\lambda^{-1}$ in the spacetime. 
The early region $u<0$ 
is outside the null cone of $(t,r)=(0,0)$, 
and the mode functions there coincide with those in full Minkowski space. 
The mode functions in the late region $u > \lambda^{-1}$ carry a memory of 
the field evolution that occurred over the intermediate region 
$0<u<\lambda^{-1}$. { The infinite contributions to the energy at $r=0$ (positive infinity) and $r=t$ (negative infinity) are shown as heavy dashed lines.} The spacelike hypersurface 
$t = T > \lambda^{-1}$, shown as a short dashed line, intersects all three regions.
%The wall starts to evolve at $(t,r)=(0,0)$ and becomes a 
%fully formed Dirichlet wall at $(t,r) = (\frac{1}{\lambda},0)$. 
%The horizontal dashed line 
%represents $t=T$. Integrating $y$ over $(0,\lambda T)$ means 
%integrating along the horizontal line from $r=T$ to $r=0$. 
%Taking only $y\in(0,1)$ 
%means including only the segment with $r$ between 
%$(T,T-\frac{1}{\lambda}$). 
\label{trplane-fig}}
\end{figure}

\section{Quantised field\label{sec:quantised}}

\subsection{Field operator and the Fock vacuum}

We quantise the field by using for the spherically symmetric sector the 
mode functions found in Section 
\ref{sec:classical} and treating the nonzero angular momentum sectors as in ordinary Minkowski space. 
As we are interested in the effects due to the evolving boundary condition, 
compared with a field in ordinary Minkowski space, 
we write out only the expressions for the spherically symmetric sector. 

We expand the spherically symmetric sector of the quantised field as 
\begin{align}
\phi = 
\int^{\infty}_{0} \bigl( a_k \phi_k+a^\dagger_k \overline{\phi}_k \bigr) 
\, dk
\,,
\label{eq:phi-op}
\end{align}
where the annihilation and creation operators have the commutators 
$\bigl[a_k,a^\dagger_{k'}\bigr] = \delta(k-k')$. 
By the normalisation of the mode functions, 
this gives the field and its time derivative the correct equal-time commutator. 
We denote by $|0\rangle$ the state that is annihilated by all 
$a_k$ and by all the annihilation operators of the nonzero angular momentum sectors. 
In the region $u< 0$, $|0\rangle$ coincides with the usual Minkowski vacuum, 
which we denote by~$|0_M\rangle$.

\subsection{Energy density\label{subsec:energydensity}}

In the Lorentz frame of the 
metric~\eqref{eq:Mink-spherical-metric}, the 
energy density of the classical scalar field is 
given in terms of the energy-momentum tensor by 
\begin{align}
T_{00} = T_{uu} + T_{vv} + 2 T_{uv}
\,,
\end{align}
where \cite{birrell-davies} 
\begin{subequations}
\label{eq:T-components}
\begin{align}
T_{uu} & =(\partial_u \phi)^2
\,,
\\
T_{vv} & =(\partial_v \phi)^2 
\,, 
\\
T_{uv} & = T_{vu} = \frac{1}{4 r^2} 
\left[(\partial_\theta\phi)^2 + (\sin\theta)^{-2}
(\partial_\varphi\phi)^2\right]
\,, 
\label{eq:T-uv}
\end{align}
\end{subequations}
and we have taken the scalar field to be minimally coupled. 
To obtain the renormalised energy 
density of the quantised field in the state~$|0\rangle$, 
$\langle T_{00} \rangle := \langle0|T_{00}|0\rangle_{\text{ren}}$, 
we point-split the expressions in~\eqref{eq:T-components}, 
take the expectation value in~$|0\rangle$, 
renormalise by subtracting the corresponding expectation value in~$|0_M\rangle$,
and finally take the coincidence limit. 
As $|0\rangle$ and $|0_M\rangle$ differ only in the spherically symmetric sector, 
the derivatives in \eqref{eq:T-uv} show that 
$\langle T_{uv}\rangle=0$, and we find 
\begin{align}
\langle T_{00} \rangle = 
\lim\limits_{\substack{u_1,u_2\rightarrow u\\v_1,v_2\rightarrow v}}
(\partial_{u_1}\partial_{u_2}+\partial_{v_1}\partial_{v_2})
\Big\lbrack
\langle 0|\phi(1)\phi(2)|0\rangle-\langle 0_M|\phi(1)\phi(2)|0_M\rangle
\Big\rbrack\,,
\label{stress-energy}
\end{align}
where $\phi$ now stands for the spherically 
symmetric quantum field~\eqref{eq:phi-op}. 

To evaluate~\eqref{stress-energy}, we write 
$\phi$ in terms of $f$ as in~\eqref{eq:phi-scaledto-f}. 
Recalling that $r = (v-u)/2$, 
this gives 
\begin{align}
\langle T_{00} \rangle 
=
\frac{1}{4\pi}
\left[ \frac{\langle(\partial_u f)^2\rangle}{r^2}
  + \frac{\langle(\partial_v f)^2\rangle}{r^2}
  + \frac{\langle f(\partial_u f-\partial_v f)\rangle
      + \langle (\partial_u f-\partial_v f)f\rangle  }{2r^3}
%+ \frac{\langle f(\partial_u f-\partial_v f)\rangle+c.c.}{2r^3}
+ \frac{\langle f^2\rangle}{2r^4}
\right] 
\,.
\label{T00all-raw}
\end{align}
By~\eqref{eq:phi-scaledto-f}, 
\eqref{eq:phisubk-modes}
and~\eqref{eq:phi-op}, 
$f$ has the expansion
\begin{align}
f = \int^{\infty}_{0} \bigl( a_k U_k+a^\dagger_k \overline{U}_k \bigr) 
\, dk
\,. 
\label{eq:f-op}
\end{align}
From \eqref{eq:Umode-ansatz}, 
\eqref{eq:E-intermsof-R}
and \eqref{eq:f-op} we obtain 
for $\langle T_{00} \rangle$
the final expression 
\begin{align}
\langle T_{00} \rangle 
& =
\frac{\lambda^2}{16 \pi^2 r^2} \int_0^\infty \frac{dK}{K}
\left[ \bigl|R'_K\bigl(\lambda (t-r)\bigr) \bigr|^2 -K^2 \right]
- \frac{1}{32 \pi^2 r^2}
\, 
\frac{\partial}{\partial r}
\! \left(
\frac{\mathcal{G}_\lambda(t,r)}{r} 
\right)
\,,
\label{T00all-final}
\end{align}
where the prime on $R_K$ 
denotes the derivative with respect to the argument and 
\begin{align}
\mathcal{G}_\lambda(t,r)
& = 
\int_0^\infty 
\frac{dK}{K}
\Bigl[ 
\bigl|R_K\bigl(\lambda (t-r)\bigr) \bigr|^2 
+ 2 \cos(2 K \lambda r) - 1 
\notag
\\[1ex]
& \hspace{14ex}
+ R_K\bigl(\lambda (t-r)\bigr) e^{i K \lambda (t+r)}
+ \overline{R_K\bigl(\lambda (t-r)\bigr)} e^{-i K \lambda (t+r)}
\Bigr]
\,. 
\label{eq:calG-def}
\end{align}
The first term in \eqref{T00all-final} 
comes from the first term in~\eqref{T00all-raw}, 
the second term in \eqref{T00all-final} 
comes from the last two terms in~\eqref{T00all-raw}, 
and the second term in \eqref{T00all-raw} vanishes. 
We note in passing that $\mathcal{G}_\lambda$ is related to the renormalised
vacuum polarisation $\langle \phi^2 \rangle$ by 
\begin{align}
\langle \phi^2 \rangle = \frac{\mathcal{G}_\lambda(t,r)}{16 \pi^2 r^2}
\,. 
\end{align}

\subsection{Energy density in the early, 
late and intermediate regions\label{subsec:T00-regions}}

We consider $\langle T_{00} \rangle$ separately 
in the early region, $t < r$, in the 
late region, 
$t > r+\lambda^{-1}$, 
and in the intermediate region, 
$r \le t \le r+\lambda^{-1}$. 

In the early region, $t < r$, 
$|0\rangle$ coincides with~$|0_M\rangle$, and 
$\langle T_{00} \rangle$ vanishes. 
 {This can be seen
immediately from~\eqref{stress-energy}, 
and also by substituting \eqref{eq:R-sol-alt} into  
\eqref{T00all-final} and~\eqref{eq:calG-def}.} 

In the late region, $t > r+\lambda^{-1}$, 
the first term in \eqref{T00all-final} vanishes. 
We show in Appendix \ref{app:T00-late} that $\langle T_{00} \rangle$
is  {a pointwise well defined function, it has 
dependence on both $t$ and~$r$, it is continuous, 
and it has the asymptotic forms} 
\begin{subequations}
\begin{align}
\langle T_{00} \rangle 
& \sim 
\frac{\ln t}{4 \pi^2 r^4}
\ \ 
\text{as $t\to\infty$ with $r$ fixed}, 
\label{eq:T00-late-larget-asymptotics}
\\[1ex]
\langle T_{00} \rangle 
& \sim - \frac{\ln r}{8 \pi^2 r^4}
\ \ 
\text{as $r\to0$ with $t$ fixed}. 
\label{eq:T00-late-smallr-asymptotics}
\end{align}
\end{subequations}
On the hypersurface of $t = T = \text{constant}$ with $T>\lambda^{-1}$ 
(see Figure~\ref{trplane-fig}), 
every ball of radius less than $T - \lambda^{-1}$ contains 
hence an infinite total energy, 
due to positive $\langle T_{00}\rangle$ that diverges as $r\to0$. 

In the intermediate region, $r \le t \le r+\lambda^{-1}$, 
we show in Appendix \ref{app:T00-intermediate} that 
$\langle T_{00}\rangle$ is 
 {a pointwise well defined function, and it is 
continuous in $r$ for $t>r$.}  
Under the technical assumption that the third derivative 
of $\tan\bigl(h(y)\bigr)$ 
is non-negative for sufficiently small positive~$y$, 
we show  {in addition that 
$\langle T_{00}\rangle$ is well defined also at $t=r$ 
(where it then vanishes); 
however, due to contributions from the first 
term in~\eqref{T00all-final}, $\langle T_{00}\rangle$ 
tends to negative infinity as 
$r \to t_-$, faster than any negative multiple of $1/h\bigl(\lambda(t-r)\bigr)$. 
In particular, $\langle T_{00}\rangle$ is not continuous at $r=t$.
This implies} 
that integrating 
$\langle T_{00}\rangle$ on a hypersurface of 
$t = T = \text{constant}>0$ 
over an an arbitrarily small 
neighbourhood of $r=T$  {gives negative infinite energy.} 
The changing boundary condition creates a 
pulse of infinite negative energy travelling outwards, 
immediately to the future of the light cone of the point 
 {$(t,r) = (0,0)$} 
where the boundary condition starts to change. 

 {Combining the results of the two previous paragraphs, it}
follows that the total energy on
the hypersurface of $t = T = \text{constant}$ with $T>\lambda^{-1}$ 
is not defined, even though $\langle T_{00}\rangle$ 
exists at every point. 
Given an $r_0 \in (0, T)$, 
the total energy for $r\le r_0$ is positive infinite, 
due to a large  positive contribution from $r\to0$, 
while the total energy for $r\ge r_0$ is negative infinite, 
due to a large negative contribution from $r \to T_-$.

\subsection{Rapid boundary condition change\label{subsec:rapidchange}}

Finally, consider the limit in which the boundary 
condition changes rapidly, $\lambda\to\infty$. 
At each given point in the region $t > r$, 
$\langle T_{00} \rangle$ diverges in this limit, 
with the asymptotic form 
\begin{align}
\langle T_{00} \rangle \sim \frac{\ln \lambda}{8 \pi^2 r^4}
\,,
\label{eq:T00-late-largelambda}
\end{align}
as we show in Appendix~\ref{app:T00-late}. 
 {In the limit of rapid source creation,
$\langle T_{00}\rangle$ hence diverges
everywhere inside the light cone of the creation event.
This is in a stark contrast to the
corresponding $(1+1)$-dimensional 
wall creation, where $\langle T_{00}\rangle$ vanishes 
inside the light cone of the creation event~\cite{Brown:2015yma}.}

\section{Response of an Unruh-DeWitt detector\label{sec:UDWdetector}} 

In this section we consider an inertial 
Unruh-DeWitt (UDW) detector \cite{Unruh:1976db,DeWitt:1979}
at a fixed spatial location. 

%We wish to know what happens 
%to the detector's response as $\lambda\to\infty$. 

We consider a detector that is coupled linearly to the quantum field. 
Within first-order perturbation theory, the probability of the 
detector to undergo a transition from a state with energy $0$ 
to a state with energy $\omega$ is proportional to the response function, given by 
\cite{Unruh:1976db,DeWitt:1979,birrell-davies,wald-smallbook} 
\begin{align}
\mathcal{F}(\omega) =
\int_{-\infty}^{\infty}dt_1\int_{-\infty}^{\infty}dt_2 
\, e^{-i\omega(t_1-t_2)}
\, \chi(t_1)\chi(t_2) \, \mathcal{W}(t_1, t_2)
\,,
\label{eq:F-genformula}
\end{align}
where the smooth real-valued switching function $\chi$ specifies how the 
detector's interaction with the field is turned on and off, and 
$\mathcal{W}$ is the pull-back of the field's Wightman function 
to the detector's worldline. In the Minkowski vacuum~$|0_M\rangle$, 
we have \cite{birrell-davies}
\begin{align}
\mathcal{W}_{|0_M\rangle}(t_1, t_2) = - \frac{1}{4\pi^2 {(t_1 - t_2 - i \epsilon)}^2}
\,,
\end{align}
where the limit $\epsilon\to0_+$ is implied and encodes 
the distributional part of~$\mathcal{W}$, 
and from \eqref{eq:F-genformula} we obtain 
\cite{Louko:2014aba,Satz:2006kb,Louko:2007mu} 
\begin{align}
\mathcal{F}_{|0_M\rangle}(\omega)
&=
-\frac{\omega \Theta(-\omega)}{2\pi} 
\int_{-\infty}^{\infty} d u \, {[\chi(u)]}^2
\notag
\\[1ex]
&\hspace{3ex}
+ 
\frac{1}{2\pi^2}
\int^{\infty}_{0} 
ds \, 
\frac{\cos(\omega s)}{s^2} 
\int_{-\infty}^{\infty} d u \, 
\chi(u) [\chi(u) - \chi(u-s)] 
\ , 
% \label{eq:F1-0M-inert}
\end{align}
where $\Theta$ is the Heaviside function. Denoting by $\mathcal{F}_{|0\rangle}$ 
the response function in the state~$|0\rangle$, and setting 
$\Delta \mathcal{F} = \mathcal{F}_{|0\rangle} - \mathcal{F}_{|0_M\rangle}$, 
we then have 
\begin{align}
\Delta \mathcal{F}(\omega) =
\int_{-\infty}^{\infty}dt_1\int_{-\infty}^{\infty}dt_2 
\, e^{-i\omega(t_1-t_2)}
\, \chi(t_1)\chi(t_2) \, \Delta \mathcal{W}(t_1, t_2)
\,,
\label{eq:Delta-F-gen}
\end{align}
where 
\begin{align}
\Delta \mathcal{W}(t_1, t_2)
&= \frac{1}{4\pi r^2}
\int_0^\infty 
\Bigl(
U_k(t_1-r, t_1+r) \overline{U_k(t_2-r, t_2+r)} 
\notag
\\
&\hspace{14ex}
- 
U^M_k(t_1-r, t_1+r) \overline{U^M_k(t_2-r, t_2+r)} 
\, 
\Bigr) 
\, dk 
\,,
\label{eq:Delta-W-def}
\end{align}
$r$ is the location of the detector, 
and 
$U^M$ is as in \eqref{eq:Umode-ansatz} but with 
$E_k(u) = - e^{-iku}$ for all~$u$. 
Note that $\Delta \mathcal{W}(t_1, t_2)$ vanishes when $t_1, t_2 \le r$. 

We consider a detector that operates only in the future region, 
$t > r + \lambda^{-1}$. 
For $t_1, t_2 > r + \lambda^{-1}$, 
the integrand in \eqref{eq:Delta-W-def} 
can be rearranged and split to give 
\begin{align}
4\pi^2 r^2 
\Delta \mathcal{W}(t_1, t_2)
& = 
\int_0^\infty 
\frac{dK}{K}
\Bigl[
(1-C_K) \,e^{iK\lambda t_2} 
+ (1-\overline{C}_K) \,e^{-iK\lambda t_1}
\Bigr]
\cos(K\lambda r)
\notag 
\\[1ex]
& \hspace{2ex}
+ \int_0^\infty 
\frac{dK}{K}
\Bigl[
|C_K|^2 - \cos(K\lambda r)
\Bigr]
\notag 
\\[1ex]
& \hspace{2ex}
+ \int_0^\infty 
\frac{dK}{K}
\Bigl[
\bigl(1-e^{iK\lambda t_2} \bigr)
+ 
\bigl(1-e^{-iK\lambda t_1} \bigr)
\Bigr]
\cos(K\lambda r)
\notag 
\\[1ex]
& \hspace{2ex}
+ \int_0^\infty 
\frac{dK}{K}
\bigl(
e^{-iK\lambda (t_1 - t_2)} - 1 
\bigr)
\cos(2K\lambda r)
\notag 
\\[1ex]
& \hspace{2ex}
+ \int_0^\infty 
\frac{dK}{K}
\bigl[
\cos(2K\lambda r )
 - \cos(K\lambda r)
\bigr]
\,.
\label{eq:Delta-W-split}
\end{align}
The integrals can be evaluated 
by the formulas of Appendix~\ref{app:integrals}, 
with the result 
\begin{align}
8\pi^2 r^2 
\Delta \mathcal{W}(t_1, t_2)
& = 
H \bigl(\lambda (t_2 + r)\bigr) + H \bigl(\lambda (t_2 - r)\bigr)
+   {H \bigl(\lambda (t_1 + r)\bigr) + H \bigl(\lambda (t_1 - r)\bigr)}
\notag 
\\[1ex]
& \hspace{2ex}
+ 
\ln \! 
\left(\frac{ \lambda^2 \bigl(t_1^2-r^2\bigr)\bigl(t_2^2-r^2\bigr)}{\bigl| 4r^2 
- {(t_1-t_2)}^2 \bigr|}\right)
\notag 
\\[1ex]
& \hspace{2ex}
+ i \pi \bigl[ 
\Theta(t_2 - t_1 - 2r)
- 
\Theta(t_1 - t_2 - 2r)
\bigr]
%\notag 
%\\[1ex]
%& \hspace{2ex}
+ 2k_1 
\,,
\label{eq:Delta-W-final}
\end{align}
where the function $H$ is defined in 
Proposition \ref{prop:H-properties} 
and 
the constant $k_1$ 
is given by~\eqref{eq:k1-def}. 
Note that $\mathcal{W}(t_1, t_2)$ has singularities at $|t_1 - t_2| = 2r$, 
which is when the two points are separated by a null geodesic that bounces off the origin, 
but this singularity is only logarithmic, and $\Delta \mathcal{W}(t_1, t_2)$ 
is representable by a function. 
  {Note also that the first four terms in \eqref{eq:Delta-W-final} are real
because $t_1, t_2 > r + \lambda^{-1}$ by assumption and $H(\alpha)$ 
is real for $\alpha\ge1$ by~\eqref{eq:Hfunc-final}.} 

We consider two limits. 

First, suppose that the support of $\chi$ is contained 
in some finite interval of fixed length, 
centered at $t=t_c$, and consider the limit $t_c\to\infty$. 
By the large argument expansion of $H$ in~\eqref{eq:H-largearg}, 
the contribution from the $H$-terms in 
\eqref{eq:Delta-W-final} vanishes in this limit, 
and we have 
\begin{align}
\Delta \mathcal{F}(\omega) \sim 
\frac{(\ln t_c) \bigl|\widehat\chi(\omega)\bigr|^2}{2\pi^2 r^2}
\,,
\label{eq:Delta-F-latetime}
\end{align}
where the hat denotes the Fourier transform, 
$\widehat\chi(\omega) := \int_{-\infty}^\infty e^{-i\omega t} \, \chi(t) \, dt$. 
$\Delta \mathcal{F}$~hence diverges in this limit, proportionally to $\ln t_c$. 
This is similar to the late time divergence of 
$\langle T_{00} \rangle$~\eqref{eq:T00-late-larget-asymptotics}. 

Second, consider the limit of large~$\lambda$.
We assume that the support of $\chi$ is contained in $[r+a, \infty)$, 
where $a$ is a positive constant, 
and we take $\lambda$ large enough that $\lambda^{-1} < a$. 
By similar arguments, we find 
\begin{align}
\Delta \mathcal{F}(\omega) =
\frac{(\ln\lambda) \bigl|\widehat\chi(\omega)\bigr|^2}{4\pi^2 r^2}
\ + \ O(1)
\,. 
\label{eq:Delta-F-largelambda}
\end{align}
The $\ln\lambda$ divergence in \eqref{eq:Delta-F-largelambda} at $\lambda\to\infty$ 
is similar to the $\ln\lambda$ divergence 
of $\langle T_{00} \rangle$ in~\eqref{eq:T00-late-largelambda}.

\section{Summary and discussion\label{sec:summary}}

We have addressed the smooth and sharp creation of a pointlike 
source for a massless scalar field in $(3+1)$-dimensional Minkowski spacetime, 
implemented by 
introducing at the spatial origin a time-dependent 
boundary condition that interpolates between ordinary 
Minkowski dynamics and a Dirichlet-type boundary condition. 
We found that the process is significantly more singular than 
a corresponding creation of a wall in $(1+1)$-dimensional 
Minkowski spacetime~\cite{Brown:2015yma}. 
While $\langle T_{00}\rangle$ is well defined away from the 
source, it is unbounded from above and below: 
there is a pulse of infinite negative energy travelling outwards, 
and there is a cloud of infinite positive energy 
that lingers around the fully formed source. 
In the rapid source creation limit, 
$\langle T_{00}\rangle$ 
diverges everywhere in the  {timelike} future of 
the creation event, 
and so does the response of an Unruh-DeWitt detector that operates 
in the  {timelike} future of the creation event. 

There are two technical reasons for the differences 
between our $(3+1)$-dimensional process and the corresponding 
$(1+1)$-dimensional process analysed in~\cite{Brown:2015yma}. 
First, as our boundary condition is at a single spatial point, 
it does not divide the $(3+1)$-dimensional spacetime into two regions. 
Our boundary condition in fact resembles more closely the removal 
of a $(1+1)$-dimensional wall than its creation~\cite{Harada:2016kkq}. 
This affects both $\langle T_{00} \rangle$ 
and the response of the Unruh-DeWitt detector. 
Second, the $(3+1)$-dimensional 
$\langle T_{00} \rangle$ \eqref{T00all-raw} 
contains terms that have no counterpart in $1+1$ dimensions, 
and these additional terms are especially significant near the source. 

We emphasise that the infinite negative 
energy radiating from the evolving source 
is localised in the immediate future of the light 
cone of the point where the boundary condition starts to change, and this negative 
energy cannot be made finite by slowing down the boundary condition change. 
We have verified, adapting the methods of our Appendix \ref{app:T00-intermediate} 
and under analogous technical assumptions, that 
a similar infinite energy occurs also in the 
$(1+1)$-dimensional wall creation of 
Section 2 in~\cite{Brown:2015yma}, 
but with two qualitative differences: 
the infinite energy in \cite{Brown:2015yma} 
is localised not where the boundary condition 
starts to change but where the boundary condition 
approaches its final value, 
and the infinite energy has positive sign.  
Specifically, formula (2.17b) in
\cite{Brown:2015yma} tends to $+\infty$ as $u \to \lambda^{-1}_-$, 
so fast that the total energy 
in (2.18) and (2.19) is positive infinity. 
Formula (2.20) in \cite{Brown:2015yma} is hence not correct: 
the term denoted therein by $O(1)$ should be replaced by positive infinity. 
We suspect that similar comments may apply to formulas 
(3.7b), (3.8) 
and (3.9) in~\cite{Brown:2015yma}. Note, however, 
that the results about detector response versus total energy in \cite{Brown:2015yma} 
were obtained via the boundary condition family~(4.1), 
and they are hence not affected by the infinities that occur in (2.18)--(2.20). 

  {Our results, including the divergent negative energy near $r=t$, 
suggest} that the creation of a pointlike source in quantum 
field theory may be sufficiently singular to model the breaking of correlations 
that has been proposed to happen at the horizon of an evaporating black hole 
\cite{braunstein-et-al,Mathur:2009hf,Almheiri:2012rt,Susskind:2013tg,Almheiri:2013hfa,Hutchinson:2013kka,Harlow:2014yka}.
  {It is conceivable that the divergent negative energy near $r=t$ and the divergent 
positive energy near $r=0$ could be arranged to cancel and produce a finite total 
energy on each hypersurface of constant~$t$, 
but such a cancellation would require a nonlocal correlation between the 
regulator near $r=t$ and the regulator near $r=0$.} 

  {We note in passing that while the source creation contributes to the 
imaginary part of the Wightman function, the imaginary part of the Wightman function 
on a trajectory of constant $r$ in the late time region consists only of the terms proportional to 
$\Theta(t_2 - t_1 - 2r)$ and $\Theta(t_1 - t_2 - 2r)$ in~\eqref{eq:Delta-W-final}. 
As the imaginary part of the Wightman function is the commutator, 
this shows that the source creation does not produce a lingering violation 
of strong Huygens' principle in the late time region on a trajectory of constant~$r$. 
The source creation does hence not appear to offer opportunities 
for enhanced quantum communication of the kind
examined in~\cite{Jonsson:2014lja,Blasco:2015a,Blasco:2015b}.} 
     
Finally, we anticipate that our techniques can be adapted to address 
an evolving boundary condition on a spherical shell or ball, 
where the dynamics will be 
potentially 
more germane for modelling possible new physics in the 
spacetime of an evaporating black hole. 
In particular, will the evolving boundary condition on the spherical shell or ball
lead to diverging positive or negative energies in some regions of the spacetime?

\section*{Acknowledgments} 

We thank 
Jim Langley for providing the proof of Proposition~\ref{prop:Bn}, 
  {Eduardo Mart\'in-Mart\'inez for raising the 
question of the strong Huygens' principle violation, and}
Joel Fein\-stein and Alex Schenkel for helpful discussions. 
This work was funded in part by the Natural Sciences and Engineering 
Research Council of Canada (MEC and GK) and by 
Science and Technology Facilities Council (JL, Theory Consolidated Grant ST/J000388/1). 
For hospitality, 
GK thanks the University of Nottingham, 
and 
JL thanks the University of Winnipeg, 
the Winnipeg Institute for Theoretical Physics, 
and the Nordita 2016 
``Black Holes and Emergent Spacetime'' 
program.

\appendix

\section{Scalar Laplacian on punctured $\BbbR^n$\label{app:laplaceoperator}}

In this appendix we record relevant properties of the scalar 
Laplacian on punctured Euclidean $\BbbR^n$ with $n\ge2$. 

We use spherical coordinates in which $r$ is the radial coordinate 
and the puncture is at $r=0$. The scalar Laplacian reads 
\bea
\nabla^2 = \frac{1}{r^{n-1}}\partial_r \bigl(r^{n-1}\partial_r \bigr) + \frac{1}{r^2} \nabla^2_{S^{n-1}}
\,,
\eea
where $\nabla^2_{S^{n-1}}$ is the Laplacian on unit $S^{n-1}$. The $L_2$ inner product is 
\bea
(g_1, g_2) = \int_0^\infty r^{n-1} \, dr \int_{S^{n-1}} d\Omega \, 
\overline{g_1} g_2 \,,
\eea
where $d\Omega$ 
is the volume element on unit $S^{n-1}$. 

The scaling $g = r^{(1-n)/2} f$ maps the inner product to 
\bea
(f_1, f_2)_{sc} = \int_0^\infty dr \int_{S^{n-1}} d\Omega \, 
\overline{f_1} f_2 
\label{eq:ip-scaled}
\eea
and $\nabla^2$ to 
\bea
\nabla^2_{sc} = \partial_r^2  - \frac{(n-1)(n-3)}{4r^2} + \frac{1}{r^2} \nabla^2_{S^{n-1}}
\,. 
\label{eq:lap-scaled}
\eea
After decomposition into spherical harmonics, 
$\nabla^2_{sc}$ reduces for each harmonic to the operator $\partial_r^2 - a/r^2$, 
where $a \ge -1/4$, and the inner product 
$(\,\cdot\, , \,\cdot \,)_{sc}$ reduces to the standard $L_2$ inner 
product on the positive half-line. The self-adjoint extensions of $\nabla^2_{sc}$ 
for each harmonic can hence be analysed by standard 
methods \cite{reed-simonII,blabk} (for a pedagogical introduction see~\cite{bonneau}), 
and the outcomes are summarised in~\cite{Kunstatter:2008qx}. 
The self-adjoint extension is unique except for $a = -1/4$, which occurs in the 
spherically symmetric sector for $n=2$, and for $a=0$, 
which occurs in the spherically symmetric sector for $n=3$. In each of these two cases 
there is a $U(1)$ family
of self-adjoint extensions, characterised by a boundary condition at the origin. 

In the $n=3$ spherically symmetric sector, the boundary condition 
at the origin is
\begin{align}
\cos\theta\lim_{r\to0} f(r)
= 
L \sin\theta \lim_{r\to0} f'(r)
\,,
\end{align}
where $L$ is a positive constant of dimension length, introduced for dimensional convenience, 
and $\theta \in [0, \pi)$ is the parameter that specifies the extension. 
For $\theta \in [0, \pi/2]$ the spectrum consists of the negative continuum, while for 
$\theta \in (\pi/2,\pi)$ there is also one proper eigenvalue, 
which is positive and nondegenerate. 
The case $\theta=0$ reduces to the essentially self-adjoint 
operator $\nabla^2$ on $L_2(\BbbR^3)$.

\section{Mode function regularity across 
$r=t$\label{app:junction-differentiability}}

In this appendix we show that 
the function $B(y)$ \eqref{eq:B-sol}
is smooth at $y=0$ and the function 
$R_K(y)$ \eqref{eq:R-sol-alt} is $C^{25}$ at $y=0$. 
This shows that the mode functions
% , given by 
% \eqref{eq:phisubk-modes}, 
% \eqref{eq:Umode-ansatz}, 
% \eqref{eq:E-intermsof-R}
% and~\eqref{eq:R-sol-full}, 
are $C^{25}$ across $r=t$.

\subsection{$B(y)$ \eqref{eq:B-sol}}

We shall show that
the function $B(y)$ \eqref{eq:B-sol}
is smooth at $y=0$. 

From \eqref{eq:B-sol} it is immediate that $B(y) \to 0$ 
as $y\to0_+$. We show below in Proposition \ref{prop:Bn} 
that $B^{(n)}(y) \to 0$ 
as $y\to0_+$ for $n \in \BbbN = \{1,2,\ldots\}$.
From this it follows by L'H\^opital and 
induction in $n$ that all derivatives of 
$B(y)$ at $y=0$ exist and vanish. 

\begin{prop}
For $n \in \BbbN$, 
$B^{(n)}(y) \to 0$ 
as $y\to0_+$. 
\label{prop:Bn} 
\end{prop}

\begin{proof}
(This proof was provided by Jim Langley.) 
Let $0<y<1$, and write $g(y) := \tan\bigl(h(y)\bigr)$, 
where $h$ was defined in Section~\ref{subsec:field-eq}. 
Note that $g(y)>0$, 
$g(y)$ and all its derivatives approach $0$ as $y\to0_+$, and from 
\eqref{eq:B-sol} we have 
\begin{align}
B(y) &= 
\exp\left(-\int_{y}^{1} \frac{dz}{g(z)} \right)
\,,
\label{eq:B-ito-g}
\\
B'(y) & = B(y)/g(y)
\,. 
\label{eq:appBident-Bprime}
\end{align}

For $n \in \BbbN$, induction gives
\begin{subequations}
\begin{align}
B^{(n)}(y) &= 
P_n(y) f_n(y)
\,, 
\\
f_n(y) &= 
\frac{B(y)}{\bigl(g(y)\bigr)^n} 
\,,
\label{eq:fn-def}
\end{align}
\end{subequations}
where each $P_n$ is a polynomial in $g$ and its derivatives. 
Since each $P_n$ is bounded as $y\to0_+$, 
it suffices to show that $f_n(y) \to 0$
as $y\to0_+$ for $n \in \BbbN$.

From \eqref{eq:fn-def}
we have 
\begin{align}
\ln\bigl(f_n(y)\bigr) =
-\left(\int_{y}^{1} \frac{dz}{g(z)} \right) 
\left(1 + \frac{n \ln\bigl(g(y)\bigr)}{\displaystyle\int_{y}^{1} 
\frac{dz}{g(z)}} \right) 
\,.
\label{eq:lnf-formula}
\end{align}
As $y\to0_+$, the first parentheses in \eqref{eq:lnf-formula} tend to~$\infty$, 
while the second parentheses tend to $1$ by L'H\^opital. 
Hence $\ln\bigl(f_n(y)\bigr) \to -\infty$ 
as $y\to0_+$, by which $f_n(y) \to 0$ as $y\to0_+$. 
\end{proof}

\subsection{$R_K(y)$ \eqref{eq:R-sol-alt}}

We shall show that the function $R_K(y)$ \eqref{eq:R-sol-alt} 
is $C^{25}$ at $y=0$. 

We write \eqref{eq:R-sol-alt} as 
\begin{align}
R_K(y) = 
\begin{cases}
{\displaystyle{- e^{-iK y}}}
& \text{for $y \le 0$}
\ , 
\\
{\displaystyle{- e^{-iK y} 
- 2 i K S_K(y)}}
& \text{for $0 < y < \infty$}
\ , 
\end{cases}
\label{eq:R-sol-aalt}
\end{align}
where $K>0$ and 
\begin{subequations}
\label{eq:J-and-S:def}
\begin{align}
S_K(y) &= J_K(y)/B(y)
\,,
\label{eq:J-and-S:def:onlyS}
\\[1ex]
J_K(y) &= \int_0^{y} {B}(z) \, e^{-iK z} \, dz 
\,. 
\end{align}
\end{subequations} 
We show below in Proposition \ref{prop:Sn-limit} that 
$S_K^{(n)}(y) \to 0$ as $y\to0_+$ for $n=0,1,2,\ldots,25$. 
This and 
\eqref{eq:R-sol-aalt} 
show that $R_K(y)$ is $C^{25}$ at $y=0$. 
For the purposes of Appendix~\ref{app:T00-intermediate}, 
we formulate Proposition \ref{prop:Sn-limit} 
for $S_K$ that is defined by 
\eqref{eq:J-and-S:def}
not just for $K>0$ but for $K\in\BbbR$. 

\begin{lem}
For $K\in\BbbR$, $0<y<1$ and $n \in \{1,2,\ldots,25\}$, we have 
\begin{align}
S_K^{(n)}(y)
= \frac{h_{K,n}(y)}{B(y) \bigl(g(y)\bigr)^n}
\,,
\label{eq:Sn-formula}
\end{align} 
where $g$ was defined above \eqref{eq:B-ito-g} 
and $h_{K,n}$ satisfies  
\begin{align}
h_{K,n}^{(k)}(y) = r_{K,n,k}(y) B(y) + s_{K,n,k}(y) J_K(y)
\ \ \text{for}
\ \ 
0 \le k \le n
\,,
\end{align} 
where each $r_{K,n,k}$ and $s_{K,n,k}$ is a polynomial in $g$, 
its derivatives 
and~$e^{-iKy}$, and $r_{K,n,n}(y) \to 0$ as $y\to0_+$. 
\label{lemma:Sn-der} 
\end{lem}

\begin{proof}
Starting from \eqref{eq:J-and-S:def} 
and using repeatedly \eqref{eq:appBident-Bprime}
and the identity 
\begin{align}
J'_K(y) &= e^{-iKy} \, B(y)
\,,
\label{eq:appBident-Jprime}
\end{align}
we have verified the claim case by case for each 
$n$ and~$k$, 
with the help of algebraic computing. 
\end{proof}

\begin{prop}
For $K\in\BbbR$ and $n \in \{0,1,2,\ldots,25\}$, 
$S_K^{(n)}(y) \to 0$ 
as $y\to0_+$. 
\label{prop:Sn-limit} 
\end{prop}

\begin{proof}
Consider~$S_K$. We use in \eqref{eq:J-and-S:def:onlyS} 
L'H\^opital with \eqref{eq:appBident-Bprime} 
and~\eqref{eq:appBident-Jprime}, obtaining 
$\lim_{y\to0_+}S(y) =  \lim_{y\to0_+} J'(y)/B'(y) 
= \lim_{y\to0_+} e^{-iK y} \, g(y) =0$.  

Consider then the derivatives of~$S_K$. 
From \eqref{eq:appBident-Bprime} we have 
\begin{align}
\frac{d}{dy} \! \left[ B(y) \bigl(g(y)\bigr)^n \right]
= B(y) \bigl(g(y)\bigr)^{n-1} \bigl(1 + n g'(y)\bigr)
\,. 
\end{align}
By Lemma~\ref{lemma:Sn-der}, we may hence evaluate 
$\lim_{y\to0_+}S_K^{(n)}(y)$ for $n\ge1$ 
by applying L'H\^opital to \eqref{eq:Sn-formula} $n$ times, 
using after the $n$th differentiation $\lim_{y\to0_+} J_K(y)/B(y) 
= \lim_{y\to0_+}S_K(y) =0$. 
\end{proof}

We stopped Lemma \ref{lemma:Sn-der} 
at $n=25$ because of computing time limitations in the case-by-case proof. 
If Lemma \ref{lemma:Sn-der} extends to 
$n \in \BbbN$, the proof of Proposition \ref{prop:Sn-limit} 
generalises to $n \in \BbbN$ and implies 
smoothness of $R_K(y)$ at $y=0$.

\section{$\langle T_{00} \rangle$ at late times\label{app:T00-late}}

In this appendix we verify the properties of $\langle T_{00} \rangle$
quoted in Sections \ref{subsec:T00-regions} and \ref{subsec:rapidchange}
in the late time region, $t > r + \lambda^{-1}$. 

Let $t > r + \lambda^{-1}$. 
From  {the last line of} \eqref{eq:R-sol-full} we see that the 
first term in \eqref{T00all-final}
vanishes. It hence suffices to consider 
$\mathcal{G}_\lambda$~\eqref{eq:calG-def}, 
which  {by the last line of \eqref{eq:R-sol-full} reduces to} 
\begin{align}
\mathcal{G}_\lambda(t,r)
& = 
4 \int_0^\infty 
\frac{dK}{K}
\Bigl[
|C_K|^2 + \cos(2K\lambda r)
- 
\left(C_K e^{i K \lambda t} + \overline{C}_K e^{-i K \lambda t} \right) 
\cos(K\lambda r)
\Bigr]
\,, 
\label{eq:calG-future}
\end{align}
where the integral is convergent 
(at large $K$ in the sense of an improper Riemann integral) 
by the properties of $C_K$ noted 
in Section~\ref{subsec:classical-modefunctions}: 
$C_K$ is smooth in~$K$, $C_0=1$, and 
$C_K \to 0$ faster than any inverse power of $K$ as $K\to\infty$. 

Rearranging the integrand in \eqref{eq:calG-future} gives 
\begin{align}
\mathcal{G}_\lambda(t,r)
& = 
4\int_0^\infty 
\frac{dK}{K}
\Bigl[
(1-C_K) \,e^{iK\lambda t} 
+ (1-\overline{C}_K) \,e^{-iK\lambda t}
\Bigr]
\cos(K\lambda r)
\notag 
\\[1ex]
& \hspace{2ex}
+ 2\int_0^\infty 
\frac{dK}{K}
\Bigl[
|C_K|^2 - \cos\bigl(K\lambda(t+r)\bigr)
\Bigr]
\notag 
\\[1ex]
& \hspace{2ex}
+ 2\int_0^\infty 
\frac{dK}{K}
\Bigl[
|C_K|^2 - \cos\bigl(K\lambda(t-r)\bigr)
\Bigr]
\notag 
\\[1ex]
& \hspace{2ex}
+ 2\int_0^\infty 
\frac{dK}{K}
\Bigl[
\cos\bigl(2K\lambda r \bigr)
 - \cos\bigl(K\lambda(t+r)\bigr)
\Bigr]
\notag 
\\[1ex]
& \hspace{2ex}
+ 2\int_0^\infty 
\frac{dK}{K}
\Bigl[
\cos\bigl(2K\lambda r \bigr)
 - \cos\bigl(K\lambda(t-r)\bigr)
\Bigr]
\,.
\label{eq:calG-break-future}
\end{align}
The integrals can be evaluated 
by the formulas of Appendix~\ref{app:integrals}, 
with the result 
\begin{align}
\mathcal{G}_\lambda(t,r)
& = 2H \bigl(\lambda (t + r)\bigr) 
+ 2H \bigl(\lambda (t - r)\bigr)
+   {2 H \bigl(\lambda (t + r)\bigr) 
+ 2 H \bigl(\lambda (t - r)\bigr)}
\notag 
\\[1ex]
& \hspace{2ex} 
+ 4 \ln\!\left(\frac{\lambda(t^2-r^2)}{r}\right)
- 4\ln2 + 4k_1 
\,, 
\label{eq:calG-future-Hsimp}
\end{align}
where the function $H$ is defined in 
Proposition~\ref{eq:Hfunc-def} and 
the constant $k_1$ 
is given by~\eqref{eq:k1-def}. 

The observations in 
Sections \ref{subsec:T00-regions} and \ref{subsec:rapidchange} 
about $\langle T_{00} \rangle$ at $t >  r + \lambda^{-1}$ 
follow from \eqref{eq:calG-future-Hsimp} 
by Proposition~\ref{prop:H-properties}.

\section{$\langle T_{00} \rangle$ at intermediate times\label{app:T00-intermediate}}

In this appendix we verify the properties of 
$\langle T_{00} \rangle$
quoted in Section \ref{subsec:T00-regions} 
in the intermediate time region, $r \le t \le r + \lambda^{-1}$. 

\subsection{Preliminaries} 

For $r <  t < r + \lambda^{-1}$, the integrals in
\eqref{eq:calG-def} and in the first term in 
\eqref{T00all-final} are convergent 
because 
\eqref{eq:R-sol-full} implies for fixed $y\in(0,1)$ the small $K$ estimates 
\begin{subequations}
\begin{align}
R_K(y) &=-1+O(K)\,,
\\
|R'_K(y)|^2 & =O\bigl(K^2\bigr)\,, 
\end{align}
\end{subequations}
and the large $K$ estimates 
\begin{subequations}
\begin{align}
R_K(y) &= e^{-iKy}\bigg[1+2\frac{B'(y)}{B(y)}\frac{1}{iK}+O\bigl(K^{-2}\bigr)\bigg]
\,,
\\
|R_K(y)|^2 &= 1+O\bigl(K^{-2}\bigr) 
\,,
\\
|R'_K(y)|^2 & =K^2+O\bigl(K^{-2}\bigr) 
\,.
\end{align}
\end{subequations}

For $t=r$, the integrands in \eqref{eq:calG-def} and in the first term of 
\eqref{T00all-final} vanish. 

For $t = r + \lambda^{-1}$, the integrand in \eqref{T00all-final} vanishes, while 
\eqref{eq:calG-def} is given by 
\eqref{eq:calG-future} with $t = r + \lambda^{-1}$, and all the steps from 
\eqref{eq:calG-future} to \eqref{eq:calG-future-Hsimp} 
still hold with $t = r + \lambda^{-1}$. 

Collecting, we see that 
$\mathcal{G}_\lambda(t,r)$ \eqref{eq:calG-def} and the first term in 
\eqref{T00all-final} are well defined everywhere in 
$r \le  t \le r + \lambda^{-1}$. 

What remains is to examine the existence and continuity of 
$\partial_r\mathcal{G}_\lambda(t,r)$, and the continuity of the 
first term in~\eqref{T00all-final}. We address each in turn. 

%we need to investigate $\partial_r\mathcal{G}_\lambda(t,r)$. 
%We shall do this below in Section~\ref{app:G-diff}. 
%
%To examine the total energy on a hypersurface of constant~$t$, 
%we need in addition to estimate the first term in \eqref{T00all-final} for 
%$r <  t < r + \lambda^{-1}$. 
%We shall do this below in Section~\ref{app:T00-estimate}. 

\subsection{$\partial_r\mathcal{G}_\lambda(t,r)$\label{app:G-diff}}

We show first that $\partial_r\mathcal{G}_\lambda(t,r)$ exists and is continuous
in $r$ for $0 < r < t$, for each positive~$t$. 
We then assume that $g'''(y)\ge0$ 
for sufficiently small positive~$y$, and show that 
$\partial_r \mathcal{G}_\lambda(t,r) \to 0$ as $r \to t_-$. 
This establishes that the second term in \eqref{T00all-final} 
exists and is continuous in~$r$. 

We introduce dimensionless variables by 
$\lambda t = \TT>0$ and $\lambda r = \TT - y$, where $0<y < \TT$. 
The quantity of interest is then 
$\mathcal{G}_\lambda\bigl(\TT/\lambda ,(\TT-y)/\lambda\bigr) = F_-(y) + F_+(y)$, 
where 
\begin{subequations}
\begin{align}
F_-(y)& = 
\int_0^1
\frac{dK}{K}
\Bigl[ 
\bigl|R_K(y) \bigr|^2 
+ 2 \cos\bigl(2 K (\TT -y)\bigr) - 1 
\notag
\\[1ex]
& \hspace{14ex}
+ R_K(y) e^{i K  (2 \TT -y)}
+ \overline{R_K(y)} e^{-i K  (2 \TT -y)}
\Bigr]
\,,
\label{eq:Fminus-def}
\\[1ex]
F_+(y)& = 
\int_1^\infty
\frac{dK}{K}
\Bigl[ 
\bigl|R_K(y) \bigr|^2 
+ 2 \cos\bigl(2 K (\TT -y)\bigr) - 1 
\notag
\\[1ex]
& \hspace{14ex}
+ R_K(y) e^{i K  (2 \TT -y)}
+ \overline{R_K(y)} e^{-i K  (2 \TT -y)}
\Bigr]
\,, 
\label{eq:Fplus-def}
\end{align}
\end{subequations}
and the notation suppresses the dependence of 
$F_\pm$ on~$\TT$. 

In $F_-$, using \eqref{eq:R-sol-aalt} gives 
\begin{align}
F_-(y)& = 
2\int_0^1
dK 
\Bigl[ 
i \left( e^{iKy} - e^{iK(2\TT-y)}\right) S_K(y)
- i \left( e^{-iKy} - e^{-iK(2\TT-y)}\right) \overline{S_K(y)}
\notag
\\
& 
\hspace{14ex}
+ 2 \bigl|S_K(y) \bigr|^2 
\Bigr]
\,. 
\end{align}
Straightforward convergence estimates show that $F_-(y)$ is $C^1$ for $y>0$, 
and estimates using Proposition 
\ref{prop:Sn-limit}
show that $F'_-(y) \to 0$ as $y\to0$. 

In $F_+$, we use the identity 
\begin{align}
R_K(y) = 
e^{-iK y} - \frac{2i}{K}
\left[
\frac{B'(y)}{B(y)} e^{-iKy} - V_K(y)
\right]
\,,
% \label{eq:R-sol-full}
\end{align}
where 
\begin{align}
V_K(y)
= \frac{1}{B(y)}
\int_0^{y} {B''}(z) \, e^{-iK z} \, dz 
\,, 
\label{eq:V-def}
\end{align}
obtained by integrating \eqref{eq:R-sol-full} 
by parts. 
This gives 
\begin{align}
F_+(y)& = 
2\int_1^\infty 
dK 
\Biggl\{ 
\frac{2}{K^3}\left(\frac{B'(y)}{B(y)}\right)^2 
+ 
\frac{2}{K}\cos\bigl(2K(\TT-y)\bigr)
+ \frac{2}{K^2}\frac{B'(y)}{B(y)} \sin\bigl(2K(\TT-y)\bigr)
\notag
\\[1ex]
& \hspace{14ex}
+ \left[
- \frac{2}{K^3}\frac{B'(y)}{B(y)} e^{iKy}
+ \frac{i}{K^2}e^{iKy}
+ \frac{i}{K^2}e^{iK(2\TT-y)}
\right]V_K(y)
\notag
\\[1ex]
& \hspace{14ex}
+ \left[
- \frac{2}{K^3}\frac{B'(y)}{B(y)} e^{-iKy}
- \frac{i}{K^2}e^{-iKy}
- \frac{i}{K^2}e^{-iK(2\TT-y)}
\right]\overline{V_K(y)}
\notag
\\[1ex]
& \hspace{14ex}
+ \frac{2}{K^3} \bigl|V_K(y) \bigr|^2 
\Biggr\}
\,, 
\label{eq:Fplus-V}
\end{align}
from which straightforward estimates show that 
$F_+(y)$ is $C^1$ for $y>0$. 

To examine $F_+(y)$ and $F'_+(y)$ as $y\to0$, 
we evaluate the integral over $K$ in~\eqref{eq:Fplus-V}. 
In the terms that do not involve $V_K$, the integral over 
$K$ produces elementary functions and the 
cosine integral~$\Ci$~\cite{dlmf}. 
In the terms that involve~$V_K$, we use~\eqref{eq:V-def}, 
we interchange the integrations as justified by 
the absolute convergence of the multiple integral, and we evaluate first the 
integral over $K$ in terms of elementary functions 
and the exponential integral $E_1$~\cite{dlmf}. 
Among the terms that ensue, several have $B'$ or~$B''$ under an integral; 
however, integration by parts reduces most of 
these terms to combinations that involve
$S_1(y)$ and $T_1(y)$, where 
\begin{align}
T_K(y)
= \frac{1}{B(y)}
\int_0^{y} B(z) \, z\,  e^{-iK z} \, dz 
\,, 
\label{eq:T-def}
\end{align}
and the small $y$ behaviour of these terms and their derivatives 
can be analysed by Proposition \ref{prop:Sn-limit} 
and its generalisations. 
 {We find that $F_+$ decomposes as 
$F_+(y) = F_{+1}(y) + F_{+2}(y)$, 
where we omit the lengthy expression for $F_{+1}(y)$ 
but just note that it satisfies 
$F_{+1}(y)\to0$ and $F'_{+1}(y)\to0$ as $y\to0$, 
while the expression for $F_{+2}(y)$
for $y<1$ reads} 
%  {[FROM HERE ON NOTATION FOR $F_{+2}(y)$ FIXED WITHOUT EXPLICIT COMMENT]}
\begin{align}
F_{+2}(y)
= \frac{4}{B^2(y)}
\int_0^y dz \, B'(z) \int_0^z dt \cos t
\, B'(z-t) \, \frac{g(z) - g(z-t)}{t}
\,. 
\label{eq:F+2-g}
\end{align}

To control $F_{+2}(y)$, we introduce the additional 
technical assumption that $g'''(y)\ge0$ 
for sufficiently small positive~$y$. 
For sufficiently small positive~$y$, 
an elementary analysis then gives for $t\in[0,y]$ the inequalities 
\begin{subequations}
\label{eq:g-ineqs}
\begin{align}
& \frac{g'(y)}{y} \le \frac{g'(y) - g'(y-t)}{t} \le g''(y)
\,,
\label{eq:g-ineq1}
\\[1ex]
& \frac{g(y)}{y} \le \frac{g(y) - g(y-t)}{t} \le g'(y)
\,, 
\label{eq:g-ineq0}
\end{align}
\end{subequations}
understood at $t=0$ in the limiting sense. 
From now on we assume $y<1$ and 
so small that \eqref{eq:g-ineqs} hold. 

Consider now $F_{+2}(y)$. 
Applying L'H\^opital in \eqref{eq:F+2-g} and using~\eqref{eq:g-ineq0}, 
we find that 
$F_{+2}(y)\to0$ as $y\to0$. 

Consider then $F'_{+2}(y)$. 
Differentiating \eqref{eq:F+2-g} gives 
\begin{align}
F'_{+2}(y)
&= \frac{4}{g(y)B^2(y)}
\Biggl[
B(y) 
\int_0^y dt \cos t
\, B'(y-t) \, \frac{g(y) - g(y-t)}{t}
\notag
\\[1ex]
& \hspace{15ex}
- 2 \int_0^y dz \, B'(z) \int_0^z dt \cos t
\, B'(z-t) \, \frac{
g(z) - g(z-t)}{t}
\Biggr]
\,. 
\label{eq:F+2diff-def}
\end{align}
For the limit of $F'_{+2}(y)$ as $y\to0$, 
L'H\^opital shows that it suffices to consider 
\begin{align}
\frac{2}{g(y)B(y)}
& 
\int_0^y dt \cos t
\biggl[
- B'(y-t) \, \frac{g(y) - g(y-t)}{t}
%\notag
%\\[1ex]
%& \hspace{12ex}
+ g(y) B''(y-t) \, \frac{g(y) - g(y-t)}{t}
\notag
\\[1ex]
& \hspace{12ex}
+ g(y) B'(y-t) \, \frac{g'(y) - g'(y-t)}{t}
\Biggr] 
\,.
\label{eq:F+2diff-LH}
\end{align}
The last term in \eqref{eq:F+2diff-LH} can be controlled 
by~\eqref{eq:g-ineq1}. 
The combination of the first two terms can be controlled by 
taking $y$ to be so small that $g'< 1$, writing $B' = g B'' /(1-g')$, 
and using \eqref{eq:g-ineq0} and the monotonicity of~$g'$. 
We find that $F'_{+2}(y)\to0$ as $y\to0$. 

 {Combining these results shows that 
$\partial_r\mathcal{G}_\lambda(t,r)$ is continuous
in $r$ for $0 < r \le t$. 
This establishes that the second term in \eqref{T00all-final} 
exists at each point and is continuous in~$r$.}

\subsection{\eqref{T00all-final} first term\label{app:T00-estimate}}

To analyse the first term in \eqref{T00all-final}, it suffices to consider 
$\tilde F(y)= \tilde F_-(y) + \tilde F_+(y)$, where 
$y>0$ and 
\begin{align}
\tilde F_-(y)
& =
\int_0^1 \frac{dK}{K}
\left[ \bigl|R'_K(y) \bigr|^2 -K^2 \right]
\,,
\\[1ex]
\tilde F_+(y)
& =
\int_1^\infty \frac{dK}{K}
\left[ \bigl|R'_K(y) \bigr|^2 -K^2 \right]
\,. 
\end{align}
We show first that $\tilde F(y)$ is continuous for $y>0$. 
We then assume that $g'''(y)\ge0$ 
for sufficiently small positive~$y$, and show that 
$\tilde F(y) \to - \infty$ as $y\to0$, faster than any negative multiple of $1/g(y)$. 

In $\tilde F_-$, we use \eqref{eq:R-sol-aalt} 
and proceed as with $F_-$~\eqref{eq:Fminus-def}. We find that 
$\tilde F_-(y)$ is continuous for $y>0$ and $\tilde F_-(y) \to 0$ as $y\to0$. 

In $\tilde F_+$, we start as with $F_+$~\eqref{eq:Fplus-def}, finding 
\begin{align}
\tilde F_+(y)& = 
2\int_1^\infty 
dK 
\Biggl\{ 
\frac{2}{K^3}\left(\frac{B'(y)}{B(y)}\right)^4
+ \frac{2}{K^3} \left(\frac{B'(y)}{B(y)}\right)^2 \bigl|V_K(y) \bigr|^2 
\notag
\\[1ex]
& \hspace{14ex}
- \frac{2}{K^3} \left(\frac{B'(y)}{B(y)}\right)^3
\left[
e^{iKy} \, V_K(y) + e^{-iKy} \, \overline{V_K(y)}
\right]
\notag
\\[1ex]
& \hspace{14ex}
+ \frac{2i}{K^2} \left(\frac{B'(y)}{B(y)}\right)^2
\left[
e^{iKy} \, V_K(y) - e^{-iKy} \, \overline{V_K(y)}
\right]
\notag
\\[1ex]
& \hspace{14ex}
- \frac{i}{K^2} 
% \left(
\frac{B'(y)}{B(y)}
% \right)
\left[
e^{iKy} \, W_K(y) - e^{-iKy} \, \overline{W_K(y)}
\right]
\Biggr\}
\,, 
\label{eq:tildeFplus-VW}
\end{align}
where $V_K$ is given by \eqref{eq:V-def}
and 
\begin{align}
W_K(y)
= \frac{1}{B(y)}
\int_0^{y} {B'''}(z) \, e^{-iK z} \, dz 
\,. 
\label{eq:W-def}
\end{align}
 {This shows that $\tilde F_+(y)$ is continuous for $y>1$.}

Proceeding as with 
\eqref{eq:Fplus-V}, and assuming $y<1$, 
we find 
%  {[FROM HERE ON NOTATION FOR 
%$\tilde F_{+1}(y)$ AND $\tilde F_{+2}(y)$ 
%FIXED WITHOUT EXPLICIT COMMENT]}
$\tilde F_+(y) = \tilde F_{+1}(y) + \tilde F_{+2}(y)$, 
where 
 {we omit the lengthy expression for $\tilde F_{+1}(y)$ 
but just note that it satisfies} 
$\tilde F_{+1}(y)\to0$ as $y\to0$, and 
\begin{align}
\tilde F_{+2}(y)
= \frac{4}{g^2(y)B^2(y)}
\left[ 
\int_0^y dz \, B'(z) J(z) \ - B(y) J(y)
\right] 
\,, 
\label{eq:Ftilde+2def}
\end{align}
where 
\begin{align}
J(y) = \int_0^y dt \cos t
\, B'(y-t) \, \frac{g(y) - g(y-t)}{t}
\,. 
\label{eq:Jbare-def}
\end{align}

No assumptions about the sign of $g'''(y)$ have been made yet. 
We now assume that $g'''(y)\ge0$ for sufficiently small positive~$y$, 
and we take $y$ to be so small that 
\eqref{eq:g-ineqs} hold, 
$\cos y \ge 1/2$, 
and $g' \le 1/2$, the last of which implies $B''>0$. 
Differentiating \eqref{eq:Jbare-def} and using~\eqref{eq:g-ineqs}, 
we then have 
$J'(y) \ge \tfrac12 B(y)/y$. 
Using \eqref{eq:Ftilde+2def}, and noting that the square brackets therein 
have the derivative $- B(y) J'(y)$, L'H\^opital hence shows that 
$g(y) \tilde F_{+2}(y) \to -\infty$ as $y\to0$. 

Collecting, these observations show that $\tilde F(y)$ is continuous for 
$y>0$, but 
 {$\tilde F(y) \to -\infty$ as $y\to0$, 
faster than any negative multiple of $1/g(y)$.}

\section{Integrals\label{app:integrals}}

In this appendix we collect results about integrals that appear in 
Section \ref{sec:UDWdetector} and Appendix~\ref{app:T00-late}. 
We recall that $C_K$ \eqref{eq:CsubK-def} is smooth in~$K$, it falls 
off at large $K$ faster than any inverse power of~$K$, and $C_0=1$. 

\begin{prop}
For $\alpha,\beta>0$, we have 
\begin{subequations}
\begin{align}
&\int_0^\infty \frac{dK}{K} \left( e^{i\alpha K} - e^{i\beta K}\right) = 
\ln(\beta/\alpha)
\,, 
\label{eq:aux-identity-1}
\\[2ex]
& \int_0^\infty \frac{dK}{K} \left( e^{i\alpha K} - e^{- i\beta K}\right) = 
\ln(\beta/\alpha) + i \pi
\,,
\label{eq:aux-identity-3}
\\[2ex]
&\int_0^\infty 
\frac{dK}{K}
\Bigl[
|C_K|^2 - \cos(\alpha K)
\Bigr]
= 
\ln\alpha + k_1
\,,
\label{eq:aux-identity-2a}
\end{align}
\end{subequations}
where   {the integrals are improper Riemann integrals,}
\begin{align}
k_1 = 
\gamma + 
\int_0^1
\frac{dK}{K}
\bigl(
|C_K|^2 - 1 \bigr)
+ 
\int_1^\infty
\frac{dK}{K}\, 
|C_K|^2
\label{eq:k1-def}
\end{align}
and $\gamma$ is Euler's constant. 
\label{prop:integrals123}
\end{prop}

\begin{proof}
In \eqref{eq:aux-identity-1} and~\eqref{eq:aux-identity-3}, 
we insert a low $K$ cutoff, 
express the integral of each term 
in terms of the exponential integral~$E_1$~\cite{dlmf}, 
and use small argument form of $E_1$ to remove the cutoff. 

In~\eqref{eq:aux-identity-2a}, we 
break the integral 
into the subintervals $0<K<1$ and $1 < K < \infty$, 
express the contributions from the subintervals 
in terms of the cosine integrals $\Cin$ and $\Ci$~\cite{dlmf}, 
and use the cosine integral identities~\cite{dlmf}. 
Note that $k_1$ is finite because of the small 
and large $K$ properties of~$C_K$. 
\end{proof}

  %{FOLLOWING THEOREM AND ITS PROOF EXPANDED}

\begin{prop}
For $\alpha>0$, let 
\begin{align}
H (\alpha) := \int_0^\infty 
\frac{dK}{K}
(1-C_K) \,e^{i\alpha K}
\,, 
\label{eq:Hfunc-def}
\end{align}
where the integral is an improper Riemann integral. 
Then 
\begin{align}
H (\alpha) = 
\begin{cases}
\displaystyle{- \int_0^1
dz \, 
\frac{B(\alpha) - B(z)}{\alpha-z}
+ \bigl(1 - B(\alpha)\bigr) \bigl( \ln(\alpha^{-1} -1) + i \pi\bigr)}
& \text{for $0<\alpha < 1$;}
\\[3ex]
\displaystyle{- \int_0^1
dz \, 
\frac{B(\alpha) - B(z)}{\alpha-z}}
& \text{for $\alpha\ge1$.}
\end{cases}
\label{eq:Hfunc-final}
\end{align}
It follows that $H$ is $C^\infty$, $H(\alpha)$ is real for $\alpha\ge1$, and $H(\alpha)$
for $\alpha>1$ has the absolutely convergent series representation 
\begin{align}
H(\alpha) = 
- \sum_{p=0}^\infty \frac{1}{\alpha^{p+1}} \int_0^1 dz \, z^p \bigl(1 - B(z) \bigr)
\,. 
\label{eq:H-largearg}
\end{align}
\label{prop:H-properties} 
\end{prop}

\begin{proof}
Consider first $\Imagpart H(\alpha)$. Taking the imaginary part of \eqref{eq:Hfunc-def}
under the integral, recalling that $\int_0^\infty dK \sin(\alpha K)/K = \pi/2$ 
(since $\alpha>0$ by assumption), and introducing a large $K$ cutoff~$M>0$, 
we have 
\begin{align}
\Imagpart H(\alpha) = 
\frac{\pi}{2} + \lim_{M\to\infty} I(M,\alpha)
\,,
\label{eq:Hfunc-impart-split}
\end{align}
where
\begin{align}
I(M,\alpha) \,\,&\!\!:= - \int_0^M
\frac{dK}{K} \int_0^1 dz \, B'(z) \sin\bigl((\alpha-z)K\bigr)
\notag
\\
& = 
- \int_0^1 dz \, B'(z)
\int_0^M \frac{dK}{K} \sin\bigl((\alpha-z)K\bigr)
\notag
\\
& = 
- \int_0^1 dz \, B'(z) \Si\bigl((\alpha-z)M\bigr)
\notag
\\
& = 
- \Si\bigl((\alpha-1)M\bigr)
- 
\int_0^1 dz \, B(z) \, \frac{\sin\bigl((\alpha-z)M\bigr)}{\alpha-z}
\notag
\\
& = 
- \Si\bigl((\alpha-1)M\bigr)
- B(\alpha) \int_0^1 dz \, \frac{\sin\bigl((\alpha-z)M\bigr)}{\alpha-z}
\notag
\\
& \hspace{3ex}+ 
\int_0^1 dz \, \frac{B(\alpha)-B(z)}{\alpha-z} \sin\bigl((\alpha-z)M\bigr)
\notag 
\\
& = 
\bigl(B(\alpha) -1 \bigr) \Si\bigl((\alpha-1)M\bigr)
- B(\alpha) \Si(\alpha M)
\notag
\\
& \hspace{3ex}
+ 
\int_0^1 dz \, \frac{B(\alpha)-B(z)}{\alpha-z} \sin\bigl((\alpha-z)M\bigr)
\,.
\label{eq:Ifunc-1}
\end{align}
The first equality in \eqref{eq:Ifunc-1} is a definition, 
the second equality comes by interchanging the integrals, 
justified by the absolute convergence of the double integral, 
and the third equality uses the definition of the 
sine integral function $\Si$~\cite{dlmf}. 
The fourth equality comes from integration by parts, the fifth 
equality by decomposing the integrand, 
and the sixth equality by using again the definition of~$\Si$. 
In the last expression in~\eqref{eq:Ifunc-1}, 
the integral term vanishes as $M\to\infty$ 
by the Riemann-Lebesgue lemma, and since $\Si(x)\to \pm \pi/2$ as $x\to\pm\infty$~\cite{dlmf},  
the other two terms show that $I(M,\alpha) \to -\pi B(\alpha) + \pi/2$ as $M\to\infty$. 
From this and \eqref{eq:Hfunc-impart-split} 
we obtain the imaginary part of~\eqref{eq:Hfunc-final}. 

Consider then $\Realpart H(\alpha)$. Taking the real part of \eqref{eq:Hfunc-def}
under the integral, we introduce both a large $K$ cutoff and a small $K$ cutoff and proceed as above, 
using now the cosine integrals $\Cin$ and $\Ci$~\cite{dlmf}. Removing the cutoffs with the help 
of the cosine integral identities \cite{dlmf} gives the real part of~\eqref{eq:Hfunc-final}. 

The smoothness of $H$ and the reality of $H(\alpha)$ for $\alpha\ge1$ are immediate from~\eqref{eq:Hfunc-final}. 
The series \eqref{eq:H-largearg} follows from \eqref{eq:Hfunc-final} by writing 
${(\alpha-z)}^{-1} = \alpha^{-1} \bigl(1-(z/\alpha)\bigr)^{-1}$ and using the geometric series. 
\end{proof}

\end{document}